\documentclass[11pt]{article}

\usepackage{amsmath,amssymb,amsthm,amsbsy}
\usepackage{mathrsfs} 
\usepackage{graphicx}	
\usepackage{fullpage}	
\usepackage{multirow} 
\usepackage{setspace} 

\usepackage{lineno}			
\usepackage{algorithm} 	
\usepackage{enumerate} 	
\usepackage{xfrac}			

\usepackage[comma]{natbib}		  

\usepackage[bookmarks=true, linkbordercolor={1 1 1}, citebordercolor={1 1 1}]{hyperref}

\theoremstyle{plain} \newtheorem{lemma}{Lemma}
\theoremstyle{plain} \newtheorem{proposition}{Proposition}
\theoremstyle{plain} 
\theoremstyle{definition} \newtheorem{corollary}{Corollary}
\theoremstyle{definition} \newtheorem{definition}{Definition}
\theoremstyle{definition} \newtheorem{assumption}{Assumption}
\theoremstyle{definition} \newtheorem{example}{Example}

\theoremstyle{plain} \newtheorem*{lemma*}{Lemma}
\theoremstyle{plain} \newtheorem*{proposition*}{Proposition}
\theoremstyle{plain} \newtheorem*{theorem*}{Theorem}
\theoremstyle{definition} \newtheorem*{corollary*}{Corollary}
\theoremstyle{definition} \newtheorem*{definition*}{Definition}
\theoremstyle{definition} \newtheorem*{assumption*}{Assumption}
\theoremstyle{definition} \newtheorem*{example*}{Example}

\newcommand{\R}{\ensuremath{\mathbb{R}}}										
 
\newcommand{\E}[1]{\mathbb{E}\left[#1\right]}								
\newcommand{\Ebig}[1]{\mathbb{E}\big[#1\big]}

\newcommand{\CE}[2]{\mathbb{E}\left[#1 \left\vert\right. #2 \right]}	
\newcommand{\Prob}[1]{\mathbb{P}\left[#1\right]}						

\newcommand{\indicator}[1]{\textbf{1}_{\left\{#1\right\}}} 	
\newcommand{\pd}[2]{\frac{\partial #1}{\partial #2}}				

\newcommand{\mb}[1]{\ensuremath{\mathbf{#1}}}
\newcommand{\mc}[1]{\ensuremath{\mathcal{#1}}}							
\newcommand{\wh}[1]{\widehat{#1}}														
\newcommand{\wt}[1]{\widetilde{#1}}													
\newcommand{\beq}[1]{\begin{equation} \label{eq:#1}}
\newcommand{\eeq}{\end{equation}}
\newcommand{\beqn}{\begin{equation*}}
\newcommand{\eeqn}{\end{equation*}}

\newcommand*\samethanks[1][\value{footnote}]{\footnotemark[#1]} 

\newcommand{\eop}{\hfill \qedsymbol}				

\newcommand{\f}[1]{h_{X_1|Y_1 = y_1}(#1)}
\newcommand{\fh}[1]{h_{X_1|Y_1 = \hat{y}_1}(#1)}
\renewcommand{\cite}[1]{\citet{#1}}

\title{\bf On Bidding with Securities: Risk Aversion and Positive Dependence\footnotetext{This work was supported by the National Science Foundation under Grant NSF ECCS 10-28464.}}

\author{
Vineet Abhishek\thanks{Department of Electrical and Computer Engineering, University of Illinois, 1308 W. Main St., Urbana, IL 61801. \newline Contact: {\tt \small abhishe1@illinois.edu} and {\tt \small b-hajek@illinois.edu}},
~Bruce Hajek\samethanks[1]
~and Steven R. Williams\thanks{(Corresponding author) Department of Economics, University of Illinois, 1407 W. Gregory Dr., Urbana, IL 61801.
\newline Contact: {\tt \small swillia3@illinois.edu}, Phone: {\tt \small +1-217-333-4516}}
}
\begin{document}
\date{}
\maketitle

\begin{abstract}
\cite{DeMarzo2005} consider auctions in which bids are selected from a completely ordered family of securities whose values are tied to the resource being auctioned.  The paper defines a notion of relative steepness of families of securities and shows that a steeper family provides greater expected revenue to the seller. Two assumptions are: the buyers are risk-neutral; the random variables through which values and signals of the buyers are realized are affiliated. We show that this revenue ranking holds for the second price auction in the case of risk-aversion. However, it does not hold if affiliation is relaxed to a less restrictive form of positive dependence, namely first order stochastic dominance (FOSD). We define the relative strong steepness of families of securities and show that it provides a necessary and sufficient condition for comparing two families in the FOSD case. All results extend to the English auction.

\vspace{10pt}
\end{abstract}

\section{Introduction} \label{sec:introduction}
Consider auctioning an asset that is a resource to be developed for profit by the winning buyer. It is common in such auctions to require bids in the form of securities whose values to the seller are tied to the eventual realized value of the asset. As an alternative to simply soliciting cash bids for the asset, for instance, a seller may require buyers to compete in terms of the equity share that the seller retains of the asset's profits. Other common securities used in bidding include debt and call options. \cite{DeMarzo2005} develop a general theory of bidding with securities in the first price and the second price auctions. Bids are selected from a completely ordered family of securities and the paper focuses on the importance of the choice of the family of securities to the seller's expected revenue. The paper defines a partial ordering of families based on the notion of \textit{steepness} (to be made precise in Section~\ref{sec:risk-aversion}) and shows that the steeper family of securities provides higher expected revenue to the seller. Two assumptions are made to prove this result: (i) buyers are \textit{risk-neutral}; (ii) the random variables through which values and signals of the buyers are realized are \textit{affiliated}. Risk neutrality is a severe restriction for a financial model. Affiliation is an extremely restrictive form of positive dependence.\footnote{\cite{deCastro2010} shows that the set of affiliated probability density functions for two random variables is the complement of an open and dense set in the space of continuous probability density functions under an appropriate topology and has zero measure under an appropriate measure.}

Our objective in this paper is to explore in the case of the second price auction the dependence of the revenue ranking of families of securities upon these two assumptions.\footnote{In addition to the second price auction, \cite{DeMarzo2005} also rank families of securities in the case of the first price auction. An additional restriction on the set of securities and the dependence of values and signals beyond affiliation is required in this analysis (i.e., the \textit{log-supermodularity} of each buyer's expected profit, which is Assumption C in the paper). Our interest in this paper is in exploring the effect of relaxing the assumption of affiliation and not restricting it further. We have not been able to carry out the analysis for the first price auction at this level of generality. We do, however, discuss the extension of our results in Section \ref{sec:eng} to the commonly used English auction, which is not considered in \cite{DeMarzo2005}.} We work with a symmetric interdependent values model on the lines of \cite{Milgrom&Weber82} and risk averse buyers. We consider two additional forms of positive dependence, namely, a \textit{monotone likelihood ratio} (\textit{MLR}) property, which is weaker than affiliation;\footnote{\cite{DeMarzo2005} assume the MLR property for the case of independent private values and affiliation for the case of interdependent values. For independent private values, the MLR property and affiliation are equivalent.} and a \textit{first order stochastic dominance} (\textit{FOSD}) property, which is weaker than the MLR property. FOSD captures the idea that the observation by a bidder of a higher signal makes larger values of the other variables more likely. The additional restriction to either MLR or affiliation is attributable to their mathematical value and is typically not motivated in any practical sense. Each of these three positive dependence assumptions has been extensively used in both auction theory and information economics.

Our main results are the following:

\begin{enumerate}[(i)]
\item 
A steeper family of securities provides higher expected revenue to the seller even with risk averse buyers and assuming that the values are positively dependent on signals in the MLR sense. We in this sense extend the result of \cite{DeMarzo2005} to the case of risk aversion and a richer informational environment.

\item
We show with an example that if the notion of positive dependence among values and signals of buyers is relaxed further from MLR to FOSD, then even for risk neutral buyers the revenue ranking of families of securities of \cite{DeMarzo2005} no longer holds.

\item
We strengthen steepness to a property that we call \textit{strong steepness} in order to rank families of securities in the case of FOSD and either risk neutral or risk averse buyers. Relative strong steepness is shown to be both necessary and sufficient for comparing two families of securities in this case: one family generates a higher expected revenue for the seller than a second family for all instances of our model satisfying FOSD if and only if it is strongly steeper than the second. 

\item
Finally, we show that the above results extend to the case of the English auction.
\end{enumerate}

\noindent It is worth emphasizing that \cite{DeMarzo2005} establish only sufficiency of relative steepness as a condition to rank two families of securities according to the revenue realized by them if affiliation is the notion of positive dependence. By contrast, we show that relative strong steepness is both necessary and sufficient for ranking two families of securities according to the expected revenue realized by them if FOSD is the notion of positive dependence. Furthermore, our proofs are more straightforward than those in\cite{DeMarzo2005} and do not require its strong regularity assumption on the probability density of return conditioned on a buyer's signal. We accomplish this mainly by exploiting the properties of concave functions, which in particular is what allows the consideration of risk averse buyers in our analysis. 

Our paper complements recent work concerning the impact of security choice on the seller's expected profit from auctions. \cite{Che&Kim2010}, \cite{Kogan&Morgan2010}, and \cite{Jun&Wolfstetter2012} study how the choice of security affects the incentives of the winning bidder in choosing either a level of investment or effort that in turn affects the expected return from the asset. The first case concerns adverse selection while the second concerns moral hazard among bidders. In each case, the ranking of securities based on the seller's net expected profit does not agree with the ranking according to relative steepness in the sense of \cite{DeMarzo2005}. None of these three papers, however, explore the effect of risk aversion or the role of the positive dependence assumption in their assessment of security bids. 

This paper is organized as follows. Section~\ref{sec:model} outlines our model, notation, and definitions. Section \ref{sec:risk-aversion} extends the revenue ranking of families of securities of \cite{DeMarzo2005} to risk averse buyers. Section~\ref{sec:pos-dep} shows that this ranking is not preserved under a more general form of positive dependence, i.e., FOSD. The revenue ranking of families of securities based on strong steepness is then presented. Section \ref{sec:eng} provides a brief overview of how the results of Sections \ref{sec:risk-aversion} and \ref{sec:pos-dep} extend to the case of the English auction. We conclude in Section \ref{sec:conclusions}.

\section{Model, Notation, and Assumptions} \label{sec:model}
Consider $N$ buyers competing for a resource that a seller wants to sell. Each buyer has a value for the resource that is unknown to him; however, each buyer has some information (\textit{signal}) about the value of the resource. The signal of a buyer is known only to him, but it may be informative to other buyers in the sense that it may improve their respective estimates of the value of the resource.

We model this by assuming that the value of the resource to a buyer $n$, denoted by $x_n$, is a realization of a nonnegative random variable $X_n$, unknown to him. This is the profit to buyer~$n$ from developing the resource in the absence of any payments to the seller but after taking into account the variable costs. A buyer $n$ privately observes a signal $y_n$ through a realization of a random variable $Y_n$ that is correlated with $(X_1, X_2, \ldots, X_N)$. A winning buyer needs to invest a fixed amount $I  \in \R$, which is the same for each buyer, to develop the resource. We allow for negative values of $I$; a negative value represents a subsidy by a third party that goes to the winner to help develop the resource. As in \cite{DeMarzo2005}, we assume that the realization of $X_n$ is observed ex-post by the seller and buyer $n$ if buyer~$n$ wins and subsequently uses the resource. The joint cumulative distribution function (CDF) of the random variables $X_n$'s and $Y_n$'s is common knowledge.  

Let $\mb{x} \triangleq (x_1, x_2, \ldots, x_N)$ denote a vector of values and let the random vector $(X_1, X_2, \ldots, X_N)$ be denoted by $\mb{X}$. A vector of signals $\mb{y}$ and the random vector $\mb{Y}$ are defined similarly.  We use the standard game theoretic notation of $\mb{x}_{-n} \triangleq (x_1, \ldots, x_{n-1}, x_{n+1}, \ldots, x_N)$, and similarly for $\mb{X}_{-n}$, $\mb{y}_{-n}$, and $\mb{Y}_{-n}$. 

Let $F_{\mb{X},\mb{Y}}(\mb{x},\mb{y})$ denote the joint CDF of $(\mb{X}, \mb{Y})$. It is assumed to have the following symmetry property:

\begin{assumption} \label{assumption:symmetry}
The joint CDF of $(X_n, Y_n, \mb{Y}_{-n})$, denoted by $F_{X_n, Y_n, \mb{Y}_{-n}}(x_n, y_n, \mb{y}_{-n})$, is identical for each $n$ and is symmetric in its last $N-1$ arguments (i.e., in the coordinates of $\mb{y}_{-n}$).
\end{assumption}

\noindent Assumption \ref{assumption:symmetry} allows for a special dependence between the value of the resource to a buyer and his own signal, while the identities of other buyers are irrelevant to him. The model reduces to the independent private values model if $(X_n, Y_n)$ is independent of $(\mb{X}_{-n},\mb{Y}_{-n})$ for all~$n$,  to the pure common value model if $X_1 = X_2 = \ldots = X_N$, and includes a continuum of interdependent value models between these two extremes. Because of Assumption~\ref{assumption:symmetry}, the subsequent assumptions and analysis are given from buyer~$1$'s viewpoint.

The set of possible values that each $X_n$ can take is assumed to be an interval $[\underline{x},\overline{x}]$ and the set of possible values that each $Y_n$ can take is assumed to be an interval $[\underline{y},\overline{y}]$. Assume that the joint probability density function (pdf) of the random vector $\mb{Y}$, denoted by $f_{\mb{Y}}(\mb{y})$, exists and is positive for all $\mb{y} \in [\underline{y},\overline{y}]^N$. By Assumption \ref{assumption:symmetry}, $f_{\mb{Y}}(\mb{y})$ is symmetric in its $N$ arguments. Define the random variable $Z_1$ as the largest among $Y_2, Y_3, \ldots, Y_N$, i.e., $Z_1 \triangleq \max\{Y_2, Y_3, \ldots, Y_N\}$; denote a realization of $Z_1$ by $z_1$.

It is commonly assumed in auction theory that the observation of a larger signal corresponds to more favorable estimates of the value of the resource. This is captured by first order stochastic dominance. The specific property that we need in our analysis of the second price auction is as follows:

\begin{definition}[FOSD] \label{defn:pd-fosd}
The random variable $X_1$ is \textit{positively dependent on the random variables $(Y_1,Z_1)$ in the first order stochastic dominance (FOSD)} sense if for any $x_1$, $1 - F_{X_1|Y_1 = y_1, Z_1 = z_1}(x_1)$ is nondecreasing in $y_1$ and $z_1$, where $F_{X_1|Y_1 = y_1, Z_1 = z_1}(x_1)$ is the CDF of $X_1$ conditioned on $Y_1 = y_1$ and $Z_1 = z_1$.\footnote{With the exception of the discussion of English auctions in Section \ref{sec:eng}, the properties ``FOSD'' and ``MLR'' in this paper specifically concern positive dependence of $X_1$ with respect to $(Y_1,Z_1).$}
\end{definition}

\noindent The following characterization of FOSD is well known:

\begin{lemma} \label{lemma:pd-fosd-eq}
FOSD is equivalent to $\E{h(X_1)|Y_1 = y_1, Z_1 = z_1}$ being nondecreasing in $y_1$ and $z_1$ for any nondecreasing function $h:\R \mapsto \R$ for which the expectation exists.
\end{lemma}

The monotone likelihood ratio property and affiliation are two more restrictive notions of positive dependence among variables that are also commonly used in auction theory. The versions that we use here are as follows:

\begin{definition}[MLR] \label{defn:pd-mlr}
Assume that for any $y_1$ and $z_1$, the pdf of $X_1$ conditioned on $Y_1 = y_1$ and $Z_1 = z_1$, denoted by $f_{X_1|Y_1 = y_1, Z_1 = z_1}(x_1)$, exists and is positive everywhere on $[\underline{x},\overline{x}]$. The random variable $X_1$ is \textit{positively dependent on the random variables $(Y_1,Z_1)$ in the monotone likelihood ratio (MLR)} sense if $f_{X_1|Y_1 = y_1, Z_1 = z_1}(x_1) / f_{X_1|Y_1 = \wh{y}_1, Z_1 = \wh{z}_1}(x_1)$ is nondecreasing in $x_1$ for any $y_1 \geq \widehat{y}_1$ and $z_1 \geq \widehat{z}_1$.
\end{definition}

\begin{definition}[Affiliation] \label{defn:pd-affiliation}
Assume that the joint pdf of $(X_1,\mb{Y})$, denoted by $f_{X_1,\mb{Y}}(x_1,\mb{y})$, exists and is positive everywhere on $[\underline{x},\overline{x}]\times[\underline{y},\overline{y}]^N$. The random variables $(X_1,\mb{Y})$ are \textit{affiliated} if
\beqn
f_{X_1,\mb{Y}}( (x_1,\mb{y}) \vee (\wh{x}_1,\wh{\mb{y}}) ) f_{X_1,\mb{Y}}( (x_1,\mb{y}) \wedge (\wh{x}_1,\wh{\mb{y}}) ) 
\geq f_{X_1,\mb{Y}}(x_1,\mb{y})f_{X_1,\mb{Y}}(\wh{x}_1,\wh{\mb{y}}),
\eeqn
for any $(x_1,\mb{y})$ and $(\wh{x}_1,\wh{\mb{y}})$ in the support of $(X_1,\mb{Y})$. Here ``$\vee$'' denotes coordinatewise maximum and ``$\wedge$'' denotes coordinatewise minimum.
\end{definition}

Under Assumption \ref{assumption:symmetry}, the following relationship between affiliation, MLR, and FOSD holds:\footnote{Assumption \ref{assumption:symmetry} is used in showing that if the random variables $(X_1,\mb{Y})$ are affiliated then so are the random variables $(X_1,Y_1,Z_1)$; see \cite{Milgrom&Weber82}. Lemma \ref{lemma:pd} then follows from the known relationship between affiliation, MLR, and FOSD; see, e.g., Chapter $1$ of \cite{Shaked&Shanthikumar06} and Appendix D of \cite{Krishna2002}.}

\begin{lemma} \label{lemma:pd}
Affiliation implies MLR and MLR implies FOSD.
\end{lemma}

Our focus is on comparing MLR and FOSD. Lemma \ref{lemma:pd} implies that results obtained by assuming FOSD hold if MLR is assumed instead, and results obtained by assuming MLR hold if affiliation is assumed instead. It is common in auction theory to justify the assumption of either affiliation or the MLR property by citing either the defining property of FOSD or the property that characterizes it in Lemma \ref{lemma:pd-fosd-eq}.\footnote{Quoting \cite{Milgrom&Weber82}:``Roughly, this (affiliation) means that a high value of one bidder's estimate makes  high values of the others' estimates more likely.'' This appealing intuition for affiliation, however, suggests the shifting of a probability distribution with the observation of a higher estimate as in first order stochastic dominance and not the inequality that defines affiliation. \cite{deCastro2010} provides some additional examples and references where affiliation is used to obtain important results in economics and finance.} The relationship in Lemma~\ref{lemma:pd}, however, does not go in the reverse direction: affiliation is strictly stronger than MLR,\footnote{In the case of second price auctions, affiliation among $X_{1}$ and $\mb{Y}$ is unnecessary; all that is needed for the analysis of \cite{Milgrom&Weber82} is affiliation among $X_{1}$, $Y_{1}$, and $Z_{1}$. Even this weaker form of affiliation, however, is strictly stronger than the MLR property we use in this paper. Affiliation among $X_{1}$, $Y_{1}$, and $ Z_{1}$ implies that an MLR ordering property holds for \textit{any} possible conditioning among these variables, e.g.,  $(X_{1}|Y_{1},Z_{1}),$ $(X_{1},Y_{1}|Z_{1}),$ $(Y_{1}|Z_{1}),$ etc.; the MLR property we use constrains only the conditioning $(X_{1} |Y_{1},Z_{1}).$   In particular, it does not require that $Y_{1}$ and $Z_{1}$ be affiliated. For example, the MLR property holds under the following assumptions: the signals $Y_{1},\ldots ,Y_{n},$ have any symmetric joint pdf; a common value $X$ is assumed (so $X=X_i$ for all $i$); $X$ is conditionally independent of $Y_{1},\ldots ,Y_{n}$  given $Y^*,$  where $Y^*$ denotes the maximum of the signal values; and the conditional distribution of  $X$ given $Y^*=y^*$ is nondecreasing (in the MLR order as a distribution for $X$)  with respect to $y^*.$ That is so because $Y^*$ is then a nondecreasing function of $(Y_1,Z_1)$  (namely, $Y^*=\max\{Y_1, Z_1\}$) and, in turn, the conditional distribution of $X$ is MLR nondecreasing in $Y^*.$} and, as discussed further in Section~\ref{sec:pos-dep}, MLR is strictly stronger than FOSD.

The buyers are assumed to be risk averse or risk neutral. Each buyer has the same von Neumann-Morgenstern utility of money, denoted by $u: \R \rightarrow \R$, which is concave (possibly linear), increasing, and normalized so that $u(0) = 0$. Henceforth, the term \textit{risk averse} includes risk neutral behavior. The seller is risk neutral. Conditioned on any $y_1$ and $z_1$, the expected utility of the resource to buyer $1$ without any payments is assumed to be positive, i.e., $\E{u(X_1 - I) | Y_1 = y_1, Z_1 = z_1} > 0$. Thus, the buyers who compete for the resource expect to make a positive profit from utilizing it.

As in \cite{DeMarzo2005}, buyers bid with securities from some ordered family. Let $\Phi \triangleq \{\phi(\cdot,b) \left|\right. b \in [\underline{b},\overline{b}]\}$ be a family of securities parametrized by $b$. A bid $b$ of buyer $1$ denotes his willingness to pay an amount $\phi(x_1,b)$ to the seller if $X_1 = x_1$. The interval $[\underline{b},\overline{b}]$ can be normalized to any arbitrary closed interval, independently of $\phi$, by translation and rescaling of the parameter~$b$ in $\phi(\cdot,b)$. It is therefore without loss of generality that we assume all families are parametrized by the same interval $[\underline{b},\overline{b}]$. The family $\Phi$ is assumed to satisfy the following conditions:

\begin{assumption}\label{assumption:admissible-security}
For any $b$, $\phi(x_1,b)$ is continuous and nondecreasing in $x_1$, and $x_1 - \phi(x_1,b)$ is nondecreasing and nonconstant in $x_1$.
\end{assumption}
\noindent Assumption \ref{assumption:admissible-security} implies that the payment made to the seller and the profit of the winning buyer are both nondecreasing in the realized value of the resource.

\begin{assumption} \label{assumption:ordered-security}
For any $y_1$ and $z_1$, 
\begin{enumerate}[(i)]
\item
$\E{u(X_1 - I - \phi(X_1,b)) | Y_1 = y_1, Z_1 = z_1}$ is continuous and decreasing in $b$, nonnegative for $b = \underline{b}$, and nonpositive for $b = \overline{b}$.
\item
$\E{\phi(X_1,b)|Y_1 = y_1, Z_1 = z_1}$ is continuous and increasing in $b$.
\end{enumerate}
\end{assumption}

\noindent Assumption \ref{assumption:ordered-security} says that the family of securities is completely ordered from the perspectives of both the winning buyer and the seller, independently of the realized signal vector. Buyers prefer lower security bids and the seller prefers higher security bids. Assumption \ref{assumption:ordered-security} is satisfied if, e.g., for any $x_1$, $\phi(x_1,b)$ is increasing in $b$. The seller uses the second price auction where the highest bidder wins and pays the security bid of the second highest bidder. As discussed in the next section, continuity together with the boundary conditions in Assumption \ref{assumption:ordered-security}(i) guarantee the existence of a pure strategy equilibrium for the second price auction. Notice that Assumption \ref{assumption:ordered-security}(i) restricts the possible values of $I$, e.g., if $\phi(x_1,b) \leq x_1$ for all $x_1$ and $b$ then $I > 0$.

Some common families of securities that satisfy Assumptions \ref{assumption:admissible-security} and \ref{assumption:ordered-security} are: \textit{cash} $\phi(x_1,b) = b, ~b \in [0,\overline{x}]$; \textit{debt} $\phi(x_1,b) = \min(x_1, b), ~b \in [0,\overline{x}]$; \textit{equity} $\phi(x_1,b) = b x_1, ~b \in [0,1-\delta]$ for some small $\delta > 0$; and \textit{call option} $\phi(x_1,b) = \max\{x_1 - \overline{x} + b, 0\}, ~b \in [0,\overline{x}]$. These families of securities are shown in Figure~\ref{fig:fig-securities}.

\begin{figure}[ht] 
\begin{center}
\includegraphics[trim=0.2in 4.5in 0.2in 0.45in, clip=true, height=1.7in]{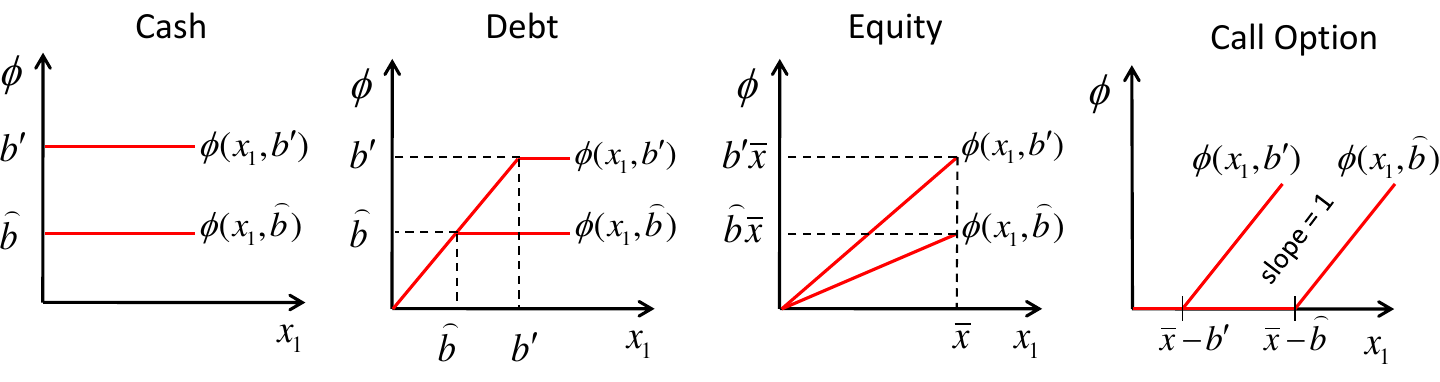}
\caption{\small \sl Plots of the families of securities cash, debt, equity, and call option for $b^{'} > \wh{b}$.\label{fig:fig-securities}} 
\end{center} 
\end{figure}

Assumptions \ref{assumption:symmetry}-\ref{assumption:ordered-security} are in place for the rest of this paper. For a comparison between two different families of securities, we use $\Psi \triangleq \{\psi(\cdot,b) \left|\right. b \in [\underline{b},\overline{b}]\}$ to denote a family of securities different from~$\Phi$.  All expectations and conditional expectations of interest are assumed to exist and be finite. 

\section{Risk Aversion} \label{sec:risk-aversion}
This section extends the result of \cite{DeMarzo2005} on revenue ranking of families of securities to risk averse buyers. In a second price auction, a buyer $n$ decides how much to bid solely based on his signal~$y_n$. We look for a symmetric equilibrium. We start by defining a function $s(y_1,z_1;\Phi)$ that will be used to characterize the bidding strategies of the buyers:
\beq{s-sp}
s(y_1, z_1;\Phi) \triangleq b : \E{u(X_1 - I - \phi(X_1,b)) | Y_1 = y_1, Z_1 = z_1} = 0.
\eeq
The uniqueness of $b$ in \eqref{eq:s-sp} follows from Assumption \ref{assumption:ordered-security}. The value $s(y_1, z_1;\Phi)$ is the security bid that makes buyer $1$ indifferent between acquiring and not acquiring the resource given his signal $y_{1}$ and the highest signal $z_{1}$ of the other buyers. Notice that the bid $s(y_1, z_1;\Phi)$ corresponds to buyer $1$'s willingness to pay an amount $\phi(x_1,s(y_1, z_1;\Phi))$ to the seller if $X_1 = x_1$. The next lemma characterizes an important property of the function $s$.

\begin{lemma} \label{lemma:s-sp-properties}
Assuming FOSD, the function $s(y_1, z_1;\Phi)$ is nondecreasing in $y_1$ and $z_1$.
\end{lemma}
\begin{proof}
Since $x_1 - I - \phi(x_1,b)$ is nondecreasing in $x_1$ by Assumption \ref{assumption:admissible-security} and $u$ is an increasing function, the claim follows immediately from Lemma~\ref{lemma:pd-fosd-eq}.
\end{proof}

To simplify the analysis in the rest of this paper, we reinforce Lemma \ref{lemma:s-sp-properties} with the following additional assumption:
\begin{assumption} \label{assumption:s-sp-monotonicity}
The family of securities and the informational environments are such that the function $s(y_1, z_1;\Phi)$ is \textit{increasing} in $y_1$.
\end{assumption}

\noindent Assumption \ref{assumption:s-sp-monotonicity} simplifies the analysis in this paper by insuring that ties among equilibrium bids (as specified in Lemma \ref{lemma:eq-s-sp} below) occur with probability zero. We therefore ignore the possibility of ties in the remainder of the paper, except in footnote \ref{tie-breaking} later in this paper. Assumption \ref{assumption:s-sp-monotonicity} is satisfied in most cases of interest; e.g., since $x_1 - \phi(x_1,b)$ is assumed to be nondecreasing and nonconstant in $x_1$, if for any $x_1 \in (\underline{x},\overline{x})$, $1 - F_{X_1|Y_1 = y_1, Z_1 = z_1}(x_1)$ is increasing in $y_1$, then Assumption \ref{assumption:s-sp-monotonicity} is automatically satisfied. The results of this paper hold without Assumption~\ref{assumption:s-sp-monotonicity} under uniform tie breaking, though the analysis is more complicated.

The next lemma characterizes an equilibrium bidding strategy for the second price auction with bids restricted to the family $\Phi$. The construction of the bidding strategy follows \cite{Milgrom&Weber82}.

\begin{lemma} \label{lemma:eq-s-sp}
Let the strategies $\beta_1, \beta_2, \ldots, \beta_N$ of the buyers be identical and defined by $\beta_n(y_n) \triangleq s(y_n, y_n; \Phi)$ for all~$n$. Assuming FOSD, the strategy vector $(\beta_1, \beta_2, \ldots, \beta_N)$ is a symmetric Bayes-Nash equilibrium of the second price auction with bids restricted to the family $\Phi$.\footnote{Our analysis and results are only for the symmetric model and for a symmetric equilibrium that is monotone, as is customary in the auction theory literature. The literature is sparse in the case of asymmetry and the results obtained in the symmetric case need not apply to the asymmetric case; see Chapter $8$ of \cite{Krishna2002} for further details.}
\end{lemma}
\begin{proof}
Assume that each buyer $n$ except buyer $1$ uses the strategy $\beta_n(y_n) = s(y_n, y_n; \Phi)$. We will show that the best response for buyer $1$ is to use the strategy $\beta_1(y_1) = s(y_1, y_1;\Phi)$.

Given $y_1$, let buyer $1$ bid $b$. Buyer $1$ wins if $b \geq \max \{s(y_n, y_n; \Phi): 2 \leq n \leq N \}$. From Lemma~\ref{lemma:s-sp-properties}, $\max\{s(y_n, y_n; \Phi): 2 \leq n \leq N\} = s(z_1, z_1; \Phi)$, where $z_1 = \max\{y_2, y_3, \ldots, y_N\}$. Thus, the expected utility of buyer $1$ is given by:
\begin{multline*}
\E{u\left(X_1 - I - \phi(X_1, s(Z_1,Z_1; \Phi)) \right) \indicator{b \geq s(Z_1,Z_1; \Phi)} | Y_1 = y_1} \\
= \E{\E{u\left(X_1 - I - \phi(X_1, s(Z_1,Z_1; \Phi)) \right) | Y_1 = y_1, Z_1} \indicator{b \geq s(Z_1,Z_1; \Phi)} | Y_1 = y_1}.
\end{multline*}
From \eqref{eq:s-sp}, $\E{u\left(X_1 - I - \phi(X_1, s(y_1,Z_1; \Phi)) \right) | Y_1 = y_1, Z_1} = 0$ and from Assumption \ref{assumption:s-sp-monotonicity}, $s(y_1,z_1; \Phi)$ is increasing in $y_1$. By Assumption \ref{assumption:ordered-security}, $\E{u\left(X_1 - I - \phi(X_1, s(Z_1,Z_1; \Phi)) \right) | Y_1 = y_1, Z_1}$ is therefore positive for $Z_1 < y_1$ and negative for $Z_1 > y_1$. The expected utility of buyer $1$ is uniquely maximized by setting $b = s(y_1, y_1; \Phi)$.
\end{proof}

Because of symmetry, the seller's expected revenue equals the expected payment made by buyer~$1$ conditioned on him winning. In the symmetric equilibrium given by Lemma \ref{lemma:eq-s-sp}, the bid of buyer $1$ is the highest if and only if his signal is the highest among all the buyers (i.e., $y_1 > z_1$). If buyer $1$ wins, his payment is determined by the second highest security bid (i.e., $s(z_1, z_1; \Phi))$. Thus, the seller's expected revenue from the second price auction with bids restricted to the family~$\Phi$ is $\E{\phi(X_1, s(Z_1, Z_1; \Phi)) | Y_1 > Z_1}$.\footnote{\label{tie-breaking}In the absence of Assumption \ref{assumption:s-sp-monotonicity}, $s(y_1, y_1; \Phi)$ need not be increasing in $y_1$ and ties can occur with positive probability. However, if we assume uniform tie breaking, the seller's expected revenue can still be shown to be $\E{\phi(X_1, s(Z_1, Z_1; \Phi)) | Y_1 > Z_1}$, implying that it is enough to assume uniform tie breaking to preserve all the results in the paper and Assumption~\ref{assumption:s-sp-monotonicity} can be dropped. To see this, notice that this is indeed the case if ties could be broken in favor of the buyer with the highest signal (which cannot be implemented because the seller cannot infer buyers' signals from their bids if $s(y_1, y_1; \Phi)$ is not invertible in $y_1$). Let $\mc{Y}$ be an interval such that $s(y_1, y_1; \Phi)$ is constant for $y_1 \in \mc{Y}$; let $r$ be this constant. The event $\{Y_1 \in \mc{Y}, Z_1 \in \mc{Y}, \text{buyer $1$ wins under uniform tie breaking} \}$ has the same probability as the event $\{Y_1 \in \mc{Y}, Z_1 \in \mc{Y}, Y_1 > Z_1\}$ and for any outcome in either of these two events, the winning buyer pays the security bid $r$.}

We next reformulate the definition of steepness from \cite{DeMarzo2005} using the concept of quasi-monotonicity, as defined below:

\begin{definition}[Quasi-monotone function] \label{defn:quasi-monotone}
A function $g(w)$ is \textit{quasi-monotone} if for any $w$ and $\wh{w}$ such that $\wh{w} < w$, if $g(\wh{w}) > 0$ then $g(w) \geq 0$. A quasi-monotone function therefore crosses zero at most once and from below.
\end{definition}

\begin{definition}[Steepness] \label{defn:steeper-security}
A family of securities $\Phi$ is \textit{steeper} than another family of securities $\Psi$ if for any $b^{'}, \widehat{b} \in [\underline{b},\overline{b}]$, $\phi(w,b^{'}) - \psi(w,\widehat{b})$ is quasi-monotone in $w$.
\end{definition}
\noindent Notice that call option is steeper than equity and debt, equity is steeper than debt, and all three of these families are steeper than cash (see Figure \ref{fig:fig-securities}).\footnote{Quasi-monotonicity is not transitive and hence steepness is not transitive. Proposition \ref{prop:sp-main-result} provides a pairwise revenue ranking for any two families of securities that are ordered under the steepness criteria. This revenue ranking, however, is transitive.}

Proposition \ref{prop:sp-main-result} below states that steepness ordering is a sufficient condition under which two different families of securities can be ranked according to the revenue they generate. The proof is in Appendix \ref{sec:proof-sp-main-result}.

\begin{proposition} \label{prop:sp-main-result}
Let $\Phi$ and $\Psi$ be two families of securities such that $\Phi$ is steeper than $\Psi$. Assuming MLR, the second price auction with bids restricted to $\Phi$ generates at least as much expected revenue for the seller as the second price auction with bids restricted to $\Psi$.
\end{proposition}

\noindent A careful review of the proof of Proposition \ref{prop:sp-main-result} shows that we in fact prove the stronger result that the expected revenue of the seller conditioned on the winning buyer's signal and the second highest signal is at least as large in the case of the steeper family of securities $\Phi$ as with the family~$\Psi$. The revenue from the steeper family thus weakly dominates in this ex-post sense, which implies that it is weakly better for the seller ex-ante as stated in the proposition. 

The following is an immediate consequence of Proposition \ref{prop:sp-main-result}:

\begin{corollary} \label{corollary:rev-rank-sp-ex}
Assuming MLR, the expected revenue from the following families of securities can be ranked as: cash $\leq$ debt $\leq$ equity $\leq$ call option.
\end{corollary}

\noindent The revenue ranking of families of securities of \cite{DeMarzo2005} is essentially Proposition~\ref{prop:sp-main-result} and Corollary \ref{corollary:rev-rank-sp-ex} in the case of risk neutral buyers and affiliated signals and values.

\section{Positive Dependence} \label{sec:pos-dep}
This section addresses the role of the positive dependence assumption in the ranking of families of securities. An example is first discussed that shows that the ranking of Proposition \ref{prop:sp-main-result} does not hold if MLR is relaxed to FOSD.\footnote{Interestingly, the example assumes independent private values among the buyers; it does not rely upon interdependence of values and the problems of inference that it creates, which is commonly the source of problems in models of trading.} The pairwise ranking of the three families of securities -- debt, equity, and call options -- is completely reversed in this example in comparison to the ranking in Corollary \ref{corollary:rev-rank-sp-ex}. If MLR is relaxed to FOSD, the relative steepness condition must be strengthened in order to rank two families of securities. This is accomplished by using the notion of \textit{strong steepness} that we define below. 

\begin{example} \label{eg:ex-fosd}
Consider two risk neutral buyers (i.e., $u(w) = w$) with independent private values. Buyer $n$'s signal $Y_n$ is uniformly distributed in the interval $[0,1]$. Conditioned on $Y_n = y_n$, the random variable $X_n$, denoting the value of buyer $n$, has the following conditional pdf:
\beq{ex-pdf}
f_{X_n|Y_n = y_n}(x_n) = \left\{ 
\begin{array}{l l}
  1 - y_n + 6x_ny_n & \quad \text{if $0 \leq x_n \leq \frac{1}{3}$,}\\
  1 & \quad \text{if $\frac{1}{3} < x_n \leq 1$.}\\
\end{array} \right.
\eeq

\begin{figure}[ht] 
\begin{center}
\includegraphics[trim=0.45in 4.25in 0.45in 0.25in, clip=true, height=2.1in]{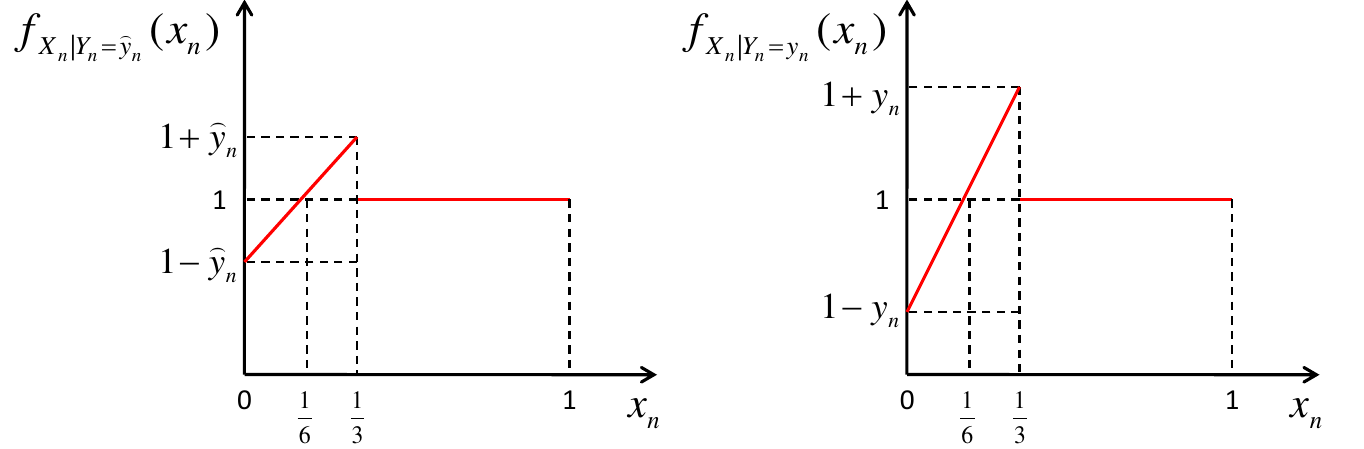}
\caption{\small \sl The pdf of $X_n$ conditioned on $Y_n = \wh{y}_n$ and conditioned on $Y_n = y_n$ for $y_n > \wh{y}_n$.\label{fig:ex1}} 
\end{center} 
\end{figure}

Figure \ref{fig:ex1} shows the plot of $f_{X_n|Y_n = y_n}(x_n)$. The pairs $(X_n,Y_n)$ are i.i.d. across the buyers. Since there are only two buyers with independent valuations, $Z_1 = Y_2$ and $F_{X_1|Y_1 = y_1, Z_1 = z_1}(x_1) = F_{X_1|Y_1 = y_1}(x_1)$. The CDF $F_{X_n|Y_n = y_n}(x_n)$ is given by:
\beq{ex-cdf}
1 - F_{X_n|Y_n = y_n}(x_n) = \left\{ 
\begin{array}{l l}
  1 - x_n + y_n(x_n - 3 x_n^2) & \quad \text{if $0 \leq x_n \leq \frac{1}{3}$,}\\
  1 - x_n & \quad \text{if $\frac{1}{3} < x_n \leq 1$.}\\
\end{array} \right.
\eeq
Since $x_n - 3 x_n^2 > 0$ for $x_n \in [0,1/3)$, $1 - F_{X_n|Y_n = y_n}(x_n)$ is increasing in $y_n$ for $x_n \in [0,1/3)$ and is constant in $y_n$ for $x_n \in [1/3,1]$. Thus, $X_n$ is positively dependent on $Y_n$ in the FOSD sense and FOSD is satisfied (in this example, $X_1$ is independent of $Z_1 = Y_2$). However, for $y_n > \wh{y}_n$, $f_{X_n|Y_n = y_n}(x_n)/f_{X_n|Y_n = \wh{y}_n}(x_n)$ fails to be nondecreasing in $x_{n}$; the ratio is strictly greater than one for $x_n \in (1/6,1/3]$ and is equal to one for $x_n \in [1/3,1]$. Thus, MLR is not satisfied.
\end{example}

Example \ref{eg:ex-fosd} highlights the distinction between MLR and FOSD in the following sense. If the random variable $X_1$ is positively dependent on the random variable $Y_1$ in the MLR sense (i.e., $f_{X_1|Y_1 = y_1}(x_1)/f_{X_1|Y_1 = \wh{y}_1}(x_1)$ is nondecreasing in $x_1$ for any $y_1 \geq \wh{y}_1$), then conditioning on a larger $Y_1$ shifts the probability distribution of $X_1$ towards the larger values of $X_1$ everywhere in the interval of possible values of $X_1$. However, if the random variable $X_1$ is positively dependent on $Y_1$ in the FOSD sense (i.e., $1 - F_{X_1 | Y_1 = y_1}(x_1)$ is nondecreasing in $y_1$ for any $x_1$), then the shift of the probability distribution towards the larger values of $X_1$ when conditioned on a larger value of $Y_1$ can be \textit{localized}; in Example \ref{eg:ex-fosd}, a larger value of $Y_1$ changes the probability distribution of $X_1$ only in the interval $[0,1/3]$, making the values in $[0,1/3]$ close to $1/3$ more likely than the values close to~$0$, while the likelihood of the values of $X_1$ in the interval $(1/3,1]$ remains unchanged. Proposition~\ref{prop:ex-fosd-ranking} below uses this difference between MLR and FOSD to show that Example \ref{eg:ex-fosd} violates the revenue ranking given by Corollary \ref{corollary:rev-rank-sp-ex}. The proof is in Appendix \ref{sec:proof-ex-fosd-ranking}.

\begin{proposition} \label{prop:ex-fosd-ranking}
For Example \ref{eg:ex-fosd}, there exists an interval of choices for the investment $I$ such that for any realization of the signal vector $(Y_1,Y_2)$, the expected revenue to the seller from the second price auction with bids restricted to debt securities is higher than the expected revenue from bids restricted to equity securities.\footnote{This ranking is robust to perturbations of the pdf $f_{X_n|Y_n = y_n}(x_n)$ in the $L^1$ sense so long as the corresponding perturbed CDF satisfies FOSD along with the other assumptions of this paper. A simple family of distributions and investment levels for which the ranking of debt over equity in Proposition \ref{prop:ex-fosd-ranking} holds can be generated as follows. For $f_{X_n|Y_n = y_n}(x_n)$ given by \eqref{eq:ex-pdf}, consider convex combinations of the form $(1-\epsilon)f_{X_n|Y_n = y_n}(x_n)+\epsilon\lambda(x_n)$ for $\epsilon\in [0,1)$ and any pdf $\lambda(x_n)$ on $[0,1]$. Such a pdf
satisfies FOSD and our other assumptions because $\lambda(x_n)$ does not depend upon $y_n$. It is straightforward to modify the proof of Proposition \ref{prop:ex-fosd-ranking} to show the existence of $\overline{\epsilon}$, $\overline{I}>0$ such that the ranking of debt over equity holds in the case of $(1-\epsilon)f_{X_n|Y_n = y_n}(x_n)+\epsilon\lambda(x_n)$ and investment $I$ for any $\left(  \varepsilon,I\right)  \in\lbrack 0,\overline{\varepsilon})\times(0,\overline{I}\rbrack$ and any pdf $\lambda(x)$.}
\end{proposition}

Recall that Corollary \ref{corollary:rev-rank-sp-ex} ranks the revenue from four families of securities in the case of MLR as: cash $\leq$ debt $\leq$ equity $\leq$ call option. Numerical computation for Example \ref{eg:ex-fosd} with investment $I = 0.2$ results in the following values for the seller's expected revenue: from cash bids $= 0.3062$; from call option $= 0.3078$; from equity $= 0.3099$; and from debt $= 0.3123$. Thus, the ranking in Example \ref{eg:ex-fosd} for $I = 0.2$ is: cash $<$ call option $<$ equity $<$ debt.  Notice that (i) cash is last in each ranking, and (ii) compared to Corollary \ref{corollary:rev-rank-sp-ex}, the relative pairwise ranking of debt, equity, and call option are reversed in this example. We show below in Corollary \ref{corollary:rev-rank-sp-refinement-ex} that point (i) holds generally in the case of FOSD, i.e., call option, equity and debt all produce a greater expected revenue for the seller than cash bids in this case. The inferiority of cash bids relative to these other securities thus generalizes from MLR to FOSD. Because the distributions that satisfy MLR form a proper subset of those that satisfy FOSD, the two rankings above show that any ranking of any pair of the three families of securities of debt, equity and call options is possible within the family of distributions that satisfy FOSD. So arriving at a definite ordering among these three families requires restricting the dependence of signals and values beyond FOSD. 

The next proposition gives a revenue ranking of families of securities that holds under FOSD with risk averse buyers. This is achieved by strengthening the steepness condition.

\begin{definition}[Strong steepness] \label{defn:strongly-steeper-security}
A family of securities $\Phi$ is \textit{strongly steeper} than another family of securities $\Psi$ if for any $b^{'}, \wh{b} \in [\underline{b},\overline{b}]$ such that $\phi(w,b^{'}) - \psi(w,\wh{b})$ assumes both negative and positive values over $w \in [\underline{x},\overline{x}]$, $\phi(w,b^{'}) - \psi(w,\wh{b})$ is nondecreasing in $w$.
\end{definition}
\noindent Notice that strong steepness implies steepness. Furthermore, debt, equity, and call option are all strongly steeper than cash.

\begin{proposition} \label{prop:rev-rank-sp-refinement}
The following statements hold:
\begin{enumerate}[(i)]
\item
Let $\Phi$ and $\Psi$ be two families of securities such that $\Phi$ is strongly steeper than~$\Psi$. Assuming FOSD, the second price auction with bids restricted to $\Phi$ generates at least as much expected revenue for the seller as the second price auction with bids restricted to~$\Psi$.
\item
Let $\Phi$ and $\Psi$ be two families of securities satisfying the following assumptions: 
\begin{enumerate}[(a)]
\item
For any $(\mb{X},\mb{Y})$ satisfying FOSD and $I \in \R$ such that Assumptions \ref{assumption:symmetry} - \ref{assumption:ordered-security} hold for both $\Phi$ and $\Psi$, $\Phi$ generates at least as much revenue as $\Psi.$
\item
For any $b \in [\underline{b},\overline{b}]$, there is a finite set $E_b$ (possibly empty) such that for any $w\in [\underline{x},\overline{x}]\backslash E_b,$ $\phi$ and $\psi$ are continuously differentiable, as functions of two variables, in a neighborhood of $(w,b).$ 
\end{enumerate}
Then the family $\Phi$ is strongly steeper than the family $\Psi$.
\end{enumerate}
\end{proposition}

\noindent The proof of Proposition \ref{prop:rev-rank-sp-refinement} is in Appendix \ref{sec:proof-rev-rank-sp-refinement}. By Assumption \ref{assumption:admissible-security}, $\phi(w,b)$ and $\psi(w,b)$ are differentiable in~$w$ almost everywhere; and by Assumption \ref{assumption:ordered-security}, $\Ebig{\phi(W,b)}$ and $\Ebig{\psi(W,b)}$ are differentiable in $b$ with positive derivatives almost everywhere. The regularity condition (b) above imposes only mild additional smoothness requirements on the securities. In particular, this assumption is satisfied by cash, debt, equity, and call option. 

As with Proposition \ref{prop:sp-main-result}, the proof of Proposition \ref{prop:rev-rank-sp-refinement}(i) establishes the stronger result that the seller's expected revenue conditional on the two highest signals is larger for the strongly steeper family of securities. The following is an immediate consequence of Proposition \ref{prop:rev-rank-sp-refinement}(i):
\begin{corollary} \label{corollary:rev-rank-sp-refinement-ex}
Assuming FOSD, the expected revenue from debt, equity, or call option are at least as large as the expected revenue from cash.
\end{corollary}

It is instructive to compare the revenue ranking of Proposition \ref{prop:rev-rank-sp-refinement}(i) to the ranking in \cite{DeMarzo2005}. Recall Example \ref{eg:ex-fosd}. As noted above, MLR shifts the distribution of $X_1$ across its support as $y_1$ increases while FOSD may only shift this distribution locally. Steepness is fundamentally a local condition that restricts how a security from one family crosses a security from another family (i.e., it crosses at most once and from below). MLR is a global notion of positive dependence that allows this local comparison of two families to determine a ranking based upon the seller's expected revenue. In moving from MLR to FOSD, however, this ranking no longer holds. Steepness is replaced in Proposition \ref{prop:rev-rank-sp-refinement}(i) by strong steepness that compares two families of securities across the entire support of $X_1$. \cite{DeMarzo2005} thus apply a local condition on families of securities together with a global condition on positive dependence in order to rank families of securities in terms of expected revenue. When the global condition on positive dependence MLR is weakened to the condition FOSD that may only bind locally, we must strengthen the comparison of the securities to a global condition that holds across the support of $X_1$ in order to be able to rank the families.

Application of Proposition \ref{prop:rev-rank-sp-refinement}(i) is illustrated further by the revenue ranking in \cite{Abhishek2011} of profit sharing securities, which are inspired by spectrum auctions in India. A fraction $\alpha\in\lbrack0,1)$ defines securities as follows. Setting $I = 0$, in the \textit{profit-loss }security, the winning buyer's payment to the seller consists of a cash bid $b$ along with an $\alpha$ share of the return $x-b$,
\[
\phi^{pl}_{\alpha}\left(  x,b\right)  = b+\alpha\left(x-b\right)  .
\]
In the \textit{profit-only }security, the winning buyer's payment to the seller consists of a cash bid along with an $\alpha$ share of the return $x-b$ when it is positive but with no additional payment when it is not,
\[
\phi^{po}_{\alpha}\left(  x,b\right)  =b+\alpha\max\left\{  0\text{, }x-b\right\}  .
\]
Let $\Phi^{pl}_{\alpha}  $ and $\Phi^{po}_{\alpha}  $ denote respectively the families of profit-loss and profit-only securities that are determined by $\alpha$ and indexed by the range of possible cash bids. It is straightforward to see that: (i) if $\alpha>\alpha^{\prime}$, then $\Phi^{pl}_{\alpha}$ is strongly steeper than $\Phi^{pl}_{\alpha^{\prime}}$ and $\Phi^{po}_{\alpha}$ is strongly steeper than $\Phi^{po}_{\alpha^{\prime}}$; (ii) for fixed $\alpha$, $\Phi^{pl}_{\alpha}$ is strongly steeper than $\Phi^{po}_{\alpha}$. Proposition \ref{prop:rev-rank-sp-refinement}(i) then implies that the seller's expected revenue in the second price auction with either profit-loss or profit-only securities is nondecreasing in the share $\alpha$, and for fixed $\alpha$ the expected revenue is weakly higher with profit and loss sharing as compared to profit only sharing.

We conclude with intuition on why a strongly steeper family of securities generates a higher expected revenue for the seller in the case of risk neutral buyers. Let $\Phi$ and $\Psi$ denote two families of securities such that $\Phi$ is strongly steeper than $\Psi$. Assume that $y_1 > z_1$ so buyer $1$ wins regardless of whether bids are from $\Phi$ or $\Psi$. Buyer $1$ in each case pays the bid of the buyer who observed signal $z_1$. His ex-post payment is equal to $\phi(x_1, s(z_{1},z_{1};\Phi))$ if bids are from~$\Phi$ and the value $X_1$ of the resource is equal to $x_1$, and the corresponding payment in the case of~$\Psi$ is $\psi(x_1,s(z_{1},z_{1};\Psi))$. In our symmetric model with risk neutral buyers, $s(z_{1},z_{1};\Phi)$ and $s(z_{1},z_{1};\Psi)$ are bids that make buyer $1$ indifferent to winning conditioned on $Y_1 = Z_1 = z_1$:
\begin{multline} \label{eq:rn-phi-psi}
\CE{\phi(X_1, s(z_1, z_1; \Phi))}{Y_1 = z_1, Z_1 = z_1} = \CE{X_1 - I}{Y_1 = z_1, Z_1 = z_1} \\
= \CE{\psi(X_1, s(z_1, z_1; \Psi))}{Y_1 = z_1, Z_1 = z_1}.
\end{multline}
The seller would thus expect to receive the same revenue from the families $\Phi$ and $\Psi$ if the highest and the second highest signals are the same, i.e.,  $Y_1=Z_1=z_1$. Buyer $1$ wins the auction, however, when his signal $Y_1 = y_1$ is greater than $z_1$. His expected payment to the seller is therefore calculated conditioned on $Y_1 = y_1 > z_1$. Intuitively, FOSD means that a larger realized signal $Y_1 = y_1$ shifts the distribution of the return $X_1$ from the resource towards its larger values. This shift increases the expected payment to the seller from the strongly steeper family of securities $\Phi$ more than that from $\Psi$ because the ex-post payment to the seller increases more rapidly as a function of $x_1$ in the case of a steeper security. Compared to $Y_1 = z_1$ for which we have the equality \eqref{eq:rn-phi-psi}, for $Y_1 = y_1 > z_1$ we have
\beq{rn-phi-psi-ineq}
\CE{\phi(X_1, s(z_1, z_1; \Phi))}{Y_1 = y_1, Z_1 = z_1} \geq \CE{\psi(X_1, s(z_1, z_1; \Psi))}{Y_1 = y_1, Z_1 = z_1}.
\eeq

We depict this intuition in Figure \ref{fig:ex2} for the case in which $\Phi$ represents equity shares and $\Psi$ represents cash. The lines represent the equilibrium bids in these two families for a given value of $z_1$. On the left is the density $f_{X_1|Y_1=z_1,Z_1=z_1}(x_1)$ of $X_1$ conditioned on $Y_1 =  Z_1 = z_1$. Relative to this density, the expected value of the payments $\phi(X_1, s(z_1, z_1; \Phi))$ and $\psi(X_1, s(z_1, z_1; \Psi))$ are equal. On the right is the density $f_{X_1|Y_1=y_1,Z_1=z_1}(x_1)$ of $X_1$ conditioned on $Y_1=y_1$ and $Z_1 = z_1$ for $y_{1}>z_{1}$. The expected value of $\phi(X_1, s(z_1, z_1; \Phi))$ exceeds the expected value of $\psi(X_1, s(z_1, z_1; \Psi))$ relative to this second density, reflecting both the shift of the density given the observation of the larger signal $Y_1=y_1>z_1$ and the relative strong steepness of the two families of securities.

\begin{figure}[t] 
\begin{center}
\includegraphics[trim=1.4in 5.75in 2in 0.70in, clip=true, height=1.5in]{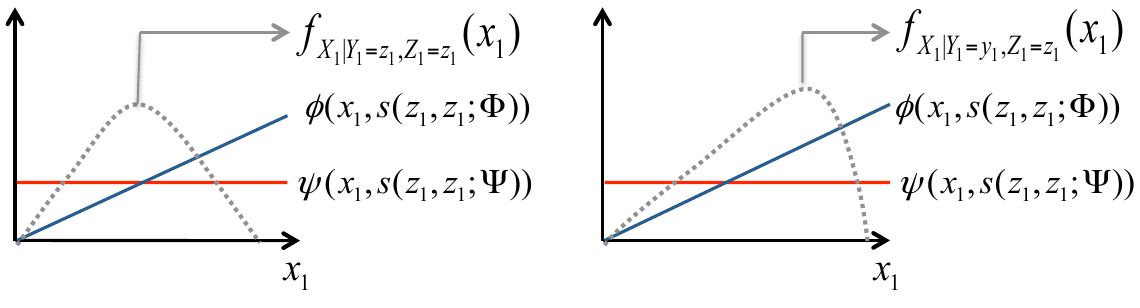}
\caption{\small \sl Payments from security bids as functions of $X_1$, overlaid on the pdf of $X_1$ conditioned on $Y_1 = Z_1 = z_1$ (left), and the pdf of $X_1$ conditioned on $Y_1 = y_1$ and $Z_1 = z_1$ for $y_1 > z_1$ (right). The family $\Phi$ represents equity shares and the family $\Psi$ represents cash payments.\label{fig:ex2}} 
\end{center} 
\end{figure}

\section{The English Auction} \label{sec:eng}
We summarize in this section the extension to the case of the English auction of all results from Sections \ref{sec:risk-aversion} and \ref{sec:pos-dep} using strengthened versions of MLR and FOSD. The English auction is an ascending price auction with a continuously increasing price. At each price level, a buyer decides whether to drop out or not. The price level and the number of active buyers are publicly known at any time. The auction ends when the second to last buyer drops out and the winner pays the price at which this happens. 

If bids are restricted to a family of securities, then the winning buyer pays the security bid at which the second to last buyer drops out. The security bids at which different buyers drop out allow the remaining buyers to infer the signals of those who have dropped out. A buyer's bidding strategy thus takes into account the number of active buyers and the inferred signals of the other buyers who have dropped out. This requires modifying MLR and FOSD such that the random variable $X_1$ is positively dependent not only on $Y_{1}$ and $Z_{1}$ but on the entire vector of signals of the other buyers, i.e., on $(Y_1, Y_2, \ldots, Y_N)$. With this modification, an equilibrium for the English auction with bids restricted to a family of securities can be characterized by following the construction of equilibrium for the English auction with cash bids in \cite{Milgrom&Weber82}.\footnote{See also \cite{Abhishek2011} concerning the English auction with a profit-sharing contract.} Using this characterization of equilibrium, all proofs in this paper extend in a straightforward manner to the case of the English auction. In particular, the analysis in Example \ref{eg:ex-fosd} of Section \ref{sec:pos-dep} holds for the English auction as well because it is strategically equivalent to the second price auction in the case of only two buyers.

\section{Conclusions} \label{sec:conclusions}
\cite{DeMarzo2005} identify the relative steepness of two families of securities as the critical factor in determining which of the two families generates the higher expected revenue for the seller in the second price and the first price auctions. For the second price auction, we first generalize this ranking to include the case of risk averse buyers. We then demonstrate the dependence of this ranking on the underlying positive dependence assumption among values and signals. An example is provided in which positive dependence is relaxed from MLR to FOSD. The pairwise revenue ranking of common families of securities -- debt, equity, and call options -- is reversed in this example from the ranking of \cite{DeMarzo2005}. The cause of this reversal is that positive dependence in the MLR sense globally restricts dependence while positive dependence in the FOSD sense may only restrict it locally; while the local condition of relative steepness is sufficient to rank families in the case of MLR, it must be strengthened in order to obtain a ranking under the less restrictive condition of FOSD. We achieve this by identifying relative strong steepness as a necessary and sufficient condition for comparing two families of securities in the case of FOSD. This result is significant because FOSD is the property that is most commonly cited in auction theory to motivate an assumption of positive dependence among values and signals. These results extend to the case of the English auction.

\appendix
\section{Proof of Proposition \ref{prop:sp-main-result}} \label{sec:proof-sp-main-result}
We start with the following definition:
\begin{definition}[Single crossing] \label{defn:single-crossing}
A function $g_1(w)$ \textit{single crosses} a function $g_2(w)$ from below if there exists $w_c$ such that $g_1(w) \leq g_2(w)$ for $w \leq w_c$ and $g_1(w) \geq g_2(w)$ for $w \geq w_c$.
\end{definition}

Lemma \ref{lemma:ohlin} and Lemma \ref{lemma:likelihood-ineq} below provide the key steps in establishing Proposition \ref{prop:sp-main-result}. 
\begin{lemma} \label{lemma:ohlin}
Let $W$ be a random variable taking values in some interval $J_1$, and let $g_i : J_1 \mapsto J_2$ for $i = 1, 2$ be nondecreasing functions with values in some interval $J_2$. Let $g_1$ single cross $g_2$ from below and $w_c$ be a crossing point. Let $h$ be any concave function. Then the following holds:
\begin{enumerate}
\item
If $\E{g_1(W)} = \E{g_2(W)}$, then $\E{h(g_1(W))} \leq \E{h(g_2(W))}$.
\item
If $\E{h(g_1(W))} = \E{h(g_2(W))}$ and $h^{'}(g_1(w_c)) > 0$, then $\E{g_1(W)} \geq \E{g_2(W)}$.
\end{enumerate}
\end{lemma}
\begin{proof}
The first claim is from Lemma $3$ of \cite{Ohlin69}. We therefore turn to the second claim, the proof of which closely follows the proof of the first claim.

Define $F_i(t) \triangleq \Prob{g_i(W) \leq t}$, $i = 1, 2$, and let $t_0 = g_1(w_c)$. Clearly, $F_1$ and $F_2$ are probability distributions. If $t < t_0$, the event $g_2(W) \leq t$ implies the event $g_1(W) \leq t$, hence $F_1(t) \geq F_2(t)$. Similarly, if $t > t_0$, the event $g_1(W) \leq t$ implies the event $g_2(W) \leq t$, hence $F_1(t) \leq F_2(t)$.

Since $h$ is concave, it is differentiable almost everywhere (in particular, the right and the left derivatives exist everywhere). Hence, $h(t) = h(t_0) + \int_{t_0}^t h^{'}(r)dr$, where $h^{'}$ can be taken as the right derivative of $h$. For $i = 1, 2$, regard $g_i(W)$ as a random variable with probability measure $F_i$. The expected value of $h(g_i(W))$ reduces as follows:

\begin{align*}
\E{h(g_i(W))} 
& = \int_{-\infty}^{\infty}h(t)dF_i(t) = h(t_0) + \int_{-\infty}^{\infty} \int_{t_0}^t h^{'}(r) dr dF_i(t),\\
& = h(t_0) - \int_{-\infty}^{t_0} \int^{t_0}_t h^{'}(r) dr dF_i(t) + \int_{t_0}^{\infty} \int_{t_0}^t h^{'}(r) dr dF_i(t), \\
& = h(t_0) - \int_{-\infty}^{t_0} \int_{-\infty}^r dF_i(t) h^{'}(r) dr + \int_{t_0}^{\infty} \int_{r}^{\infty} dF_i(t) h^{'}(r) dr, \\
& = h(t_0) - \int_{-\infty}^{t_0} F_i(r)h^{'}(r)dr + \int_{t_0}^{\infty} (1-F_i(r))h^{'}(r)dr, \\
& = h(t_0) - \int_{-\infty}^{\infty} F_i(r)h^{'}(r)dr + \int_{t_0}^{\infty} h^{'}(r)dr.
\end{align*}
Hence, 
\beq{lemma-ohlin-eq1}
\E{h(g_1(W))}  - \E{h(g_2(W))} = \int_{-\infty}^{\infty}(F_2(t) - F_1(t))h^{'}(t)dt.
\eeq
Substituting $h(t) = t$ in \eqref{eq:lemma-ohlin-eq1} implies
\beq{lemma-ohlin-eq2}
\E{g_1(W)}  - \E{g_2(W)} = \int_{-\infty}^{\infty}(F_2(t) - F_1(t))dt.
\eeq
Since $h^{'}(t)$ is nonincreasing in $t$, $F_2(t) - F_1(t) \leq 0$ for $t \leq t_0$, and $F_2(t) - F_1(t) \geq 0$ for $t \geq t_0$, we have,
\beq{lemma-ohlin-eq3}
(F_2(t) - F_1(t))h^{'}(t) \leq (F_2(t) - F_1(t))h^{'}(t_0), \quad \text{for all $t$}.
\eeq
Combining \eqref{eq:lemma-ohlin-eq1}-\eqref{eq:lemma-ohlin-eq3} results in
\beq{lemma-ohlin-eq4}
\E{h(g_1(W))}  - \E{h(g_2(W))} \leq h^{'}(t_0)\left(\E{g_1(W)}  - \E{g_2(W)}\right).
\eeq
The result then immediately follows from \eqref{eq:lemma-ohlin-eq4}.
\end{proof}

\begin{lemma} \label{lemma:likelihood-ineq}
Let a function $g$ single cross zero from below. Suppose $\E{g(X_1) | Y_1 = \widehat{y}_1, Z_1 = z_1} \geq 0$ for some $\widehat{y}_1$ and $z_1$. Assuming MLR, $\E{g(X_1) | Y_1 = y_1, Z_1 = z_1} \geq 0$ for any $y_1 > \widehat{y}_1$.
\end{lemma}
\begin{proof}
Since conditioning on $Z_1 = z_1$ plays no role in the above claim, for notational convenience define $h_{X_1|Y_1 = y_1}(x_1) \triangleq f_{X_1|Y_1 = y_1, Z_1 = z_1}(x_1)$, which omits $Z_1 = z_1$. Let $w_0$ be the point at which~$g$ crosses zero from below. Since $\f{x_1} / \fh{x_1}$ is increasing in $x_1$, $g(x_1) \leq 0$ for $x_1 \leq w_0$, and $g(x_1) \geq 0$ for $x_1 \geq w_0$, we get
\beq{likelihood-ineq-eq1}
g(x_1) \frac{\f{x_1}}{\fh{x_1}} \geq g(x_1) \frac{\f{w_0}}{\fh{w_0}} \quad \text{for all $x_1$.}
\eeq
Then, 
\begin{align*}
\E{g(X_1) | Y_1 = y_1, Z_1 = z_1} 
& = \int_{\underline{x}}^{\overline{x}} g(w)\f{w}dw, \\ 
& = \int_{\underline{x}}^{\overline{x}} g(w)\frac{\f{w}}{\fh{w}}\fh{w}dw, \\
& \geq \int_{\underline{x}}^{\overline{x}} g(w) \frac{\f{w_0}}{\fh{w_0}} \fh{w}dw, \\
& = \frac{\f{w_0}}{\fh{w_0}} \E{g(X_1) | Y_1 = \widehat{y}_1, Z_1 = z_1} \geq 0,
\end{align*}
where the first inequality is from \eqref{eq:likelihood-ineq-eq1}. This completes the proof.
\end{proof}

We can now prove Proposition \ref{prop:sp-main-result}. The seller's revenue if bids are restricted to the family $\Phi$ is $\E{\phi(X_1,s(Z_1, Z_1;\Phi)) | Y_1 > Z_1}$ and is $\E{\psi(X_1,s(Z_1, Z_1;\Psi)) | Y_1 > Z_1}$ if bids are restricted to the family $\Psi$. To prove Proposition~\ref{prop:sp-main-result}, it suffices to show that for any $y_1 > z_1$,
\beq{proof-sp-main-result-eq1}
\E{\phi(X_1,s(Z_1, Z_1;\Phi)) | Y_1 = y_1, Z_1 = z_1} \geq \E{\psi(X_1,s(Z_1, Z_1;\Psi)) | Y_1 = y_1, Z_1 = z_1}.
\eeq

From \eqref{eq:s-sp}, for any $z_1$,  
\begin{multline} \label{eq:proof-sp-main-result-eq2}
\E{u( X_1 - I - \phi(X_1,s(Z_1, Z_1;\Phi))) | Y_1 = z_1, Z_1 = z_1} = 0\\
= \E{u( X_1 - I - \psi(X_1,s(Z_1, Z_1;\Psi))) | Y_1 = z_1, Z_1  = z_1}. 
\end{multline}
If either $\phi(x_1,s(z_1,z_1;\Phi)) < \psi(x_1,s(z_1,z_1;\Psi))$ or $\phi(x_1,s(z_1,z_1;\Phi)) > \psi(x_1,s(z_1,z_1;\Psi))$ for all $x_1$ then \eqref{eq:proof-sp-main-result-eq2} would not be true. Hence, $\phi(x_1,s(z_1,z_1;\Phi))$ and $\psi(x_1,s(z_1,z_1;\Psi))$ must cross each other as functions of $x_1$. Since $\Phi$ is steeper than $\Psi$, $\phi(x_1,s(z_1,z_1;\Phi))- \psi(x_1,s(z_1,z_1;\Psi))$ single crosses zero from below. This implies that $x_1 - \psi(x_1,s(z_1 , z_1 ;\Psi))$ single crosses $x_1 - \phi(x_1,s(z_1, z_1 ;\Phi))$ from below. This, along with \eqref{eq:proof-sp-main-result-eq2}, and $u$ being concave and increasing, allow for an application of the second part of Lemma \ref{lemma:ohlin}, and results in the following inequality:
\begin{multline} \label{eq:proof-sp-main-result-eq3a}
\E{X_1 - I - \phi(X_1,s(Z_1, Z_1;\Phi)) | Y_1 = z_1, Z_1 = z_1} \\
\leq \E{X_1 - I - \psi(X_1,s(Z_1, Z_1;\Psi)) | Y_1 = z_1, Z_1 = z_1}.
\end{multline}
Hence,
\beq{proof-sp-main-result-eq3}
\E{\phi(X_1,s(Z_1, Z_1;\Phi)) - \psi(X_1,s(Z_1, Z_1;\Psi)) | Y_1 = z_1, Z_1 = z_1} \geq 0.
\eeq
Since $\phi(x_1,s(z_1,z_1;\Phi))- \psi(x_1,s(z_1,z_1;\Psi))$ single crosses zero from below, \eqref{eq:proof-sp-main-result-eq3} and Lemma \ref{lemma:likelihood-ineq} imply
\beq{proof-sp-main-result-eq4}
\E{\phi(X_1,s(Z_1, Z_1;\Phi)) - \psi(X_1,s(Z_1, Z_1;\Psi)) | Y_1 = y_1, Z_1 = z_1} \geq 0,
\eeq
for $y_1 > z_1$. This establishes \eqref{eq:proof-sp-main-result-eq1} and the proof is complete. \eop

\section{Proof of Proposition \ref{prop:ex-fosd-ranking}} \label{sec:proof-ex-fosd-ranking}
We first explain how the example was devised. Recall the statement of Lemma \ref{lemma:likelihood-ineq} from the preceding section. The proof of Proposition \ref{prop:sp-main-result} is an application of Lemma \ref{lemma:likelihood-ineq} in which for any value of $z_{1}$: (i) $g(x_1) = \phi(x_{1},s(z_{1},z_{1};\Phi))-\psi(x_{1},s(z_{1},z_{1};\Psi))$, where $\Phi$ is a steeper family than $\Psi$; (ii) $\wh{y}_{1}$ equals $z_{1}$; (iii) Lemma \ref{lemma:likelihood-ineq} is applied to derive inequality \eqref{eq:proof-sp-main-result-eq4}, which is the conclusion that the expected payment by buyer $1$ in the event that he trades is greater with the family of securities $\Phi$ than with~$\Psi$. Example \ref{eg:ex-fosd} is constructed with the goal of making this last step false so that the steeper family of securities produces a lower expected payment by buyer $1$. The key observation is that while the function $g(x_{1})$ is assumed by Lemma \ref{lemma:likelihood-ineq} to single cross zero from below, this does not preclude $g(x_{1})$ from decreasing for values of $x_{1}$ below the point at which it crosses zero. In Example \ref{eg:ex-fosd}, a larger value of $Y_1$ changes the probability density of $X_1$ only in the interval $[0,1/3]$, making the values in $[0,1/3]$ closer to $1/3$ more likely and the values near $0$ less likely, while the probability density over $[1/3,1]$ remains unaffected. If $g(x_{1})$ is decreasing over $[0,1/3]$, then conditioning on a larger value of $Y_{1}$ can decrease the expected value of $g(X_{1})$ over $[0,1]$ and thereby reverse the conclusion of Proposition \ref{prop:sp-main-result}. As we show below, this
in fact occurs for a range of values of the investment $I$ and for each realization of the signal vector $(Y_1,Y_2)$ in the case in which $\Phi$ is the equity family and~$\Psi$ is the debt family. 

We begin by choosing the investment parameter $I$ to ensure that the relevant $g(x_{1})$ in the case of debt and equity crosses zero at a value larger than $1/3$. In the case of debt securities, the optimal bid $b$ of buyer $1$ when his signal equals zero is determined by the equation
\begin{align} \label{eq:proof-ex-fosd-ranking-eq1a}
& \E{X_{1}-I-\min(X_{1},b)\left\vert Y_{1}=0\right.} = 0, \nonumber \\
\Leftrightarrow & \E{X_{1}-I\left\vert Y_{1}=0\right.} = \E{\min(X_{1},b)\left\vert Y_{1}=0\right.}.
\end{align}
With foresight to the use of $g(x_{1})$ below, we wish to ensure that the optimal bid of buyer $1$ when his signal equals zero exceeds $1/3$. The left side of \eqref{eq:proof-ex-fosd-ranking-eq1a} is decreasing in $I$ and the right side is nondecreasing in $b$. At $I=0.2$ and $b=1/3$, the left side strictly exceeds the right side.  As a consequence, we conclude that there is a value $\overline{I}>0.2$ such that for all $I<\overline{I}$, the
value of $b$ that solves \eqref{eq:proof-ex-fosd-ranking-eq1a} strictly exceeds $1/3$. We therefore
fix the investment at some value $\wt{I}\in(0,\overline{I})$. 

Consider an arbitrary realization $(\wt{y}_{1},\wt{y}_{2})$ of the signal vector such that buyer $1$ wins, i.e., $\wt{y}_{1}>\wt{y}_{2}$. Given $\wt{y}_{2}$, let $b^{d}$ denote the bid of buyer $2$ when he bids with debt securities and $b^{e}$ his bid when he bids with equity securities. It is sufficient to prove that
\beq{proof-ex-fosd-ranking-eq1}
\E{ b^e X_1 | Y_1 = \wt{y}_1} < \E{\min(X_1, b^d) | Y_1 = \wt{y}_1},
\eeq
where the left hand side denotes the seller's expected revenue given $(\wt{y}_{1},\wt{y}_{2})$ in the case of equity securities and the right hand side denotes his expected revenue in the case of debt securities. We are using here the fact that $X_{1}$ is independent of $Y_{2}$ in this example. 

Lemma \ref{lemma:eq-s-sp} states that the bids $b^{d}$ and $b^{e}$ satisfy:
\begin{align}
& \E{X_1 - \wt{I} - b^e X_1 | Y_1 = \wt{y}_2} = 0 = \E{X_1 - \wt{I} - \min(X_1, b^d) | Y_1 = \wt{y}_2}, \label{eq:proof-ex-fosd-ranking-eq2} \\
\Rightarrow & \E{b^e X_1 - \min(X_1, b^d) | Y_1 = \wt{y}_2} = 0. \label{eq:proof-ex-fosd-ranking-eq3}
\end{align}
We next apply \eqref{eq:proof-ex-fosd-ranking-eq2} to bound the bids $b^{d}$ and $b^{e}$. Since $\wt{I}>0$, it is straightforward to see that $b^{e}< 1$. Our foresight in choosing $\wt{I}$ is now useful: because $\E{X_{1}-\wt{I}-\min(X_{1},b^{d})\left\vert Y_{1}=y_{1}\right.}$ is increasing in $y_{1}$, the solution $b^{d}$ to \eqref{eq:proof-ex-fosd-ranking-eq2} is at least as large as its value at $y_{1}=0$ and so $b^{d}>1/3$.

Define $g(x_1) \triangleq b^e x_1 - \min(x_1,b^d)$. The function $g$ is decreasing in the interval $[0,b^d]$ and thus decreasing in $[0,1/3]$. From \eqref{eq:proof-ex-fosd-ranking-eq3}, $\mathbb{E}[g(X_1) | Y_1 = \tilde{y}_2] = 0$. Hence, 
\begin{multline} \label{eq:proof-ex-fosd-ranking-eq4}
\E{g(X_1) | Y_1 = \tilde{y}_1} = \E{g(X_1) | Y_1 = \tilde{y}_1} - \E{g(X_1) | Y_1 = \tilde{y}_2} \\
= \int_{0}^{1} g(w)f_{X_1|Y_1 = \tilde{y}_1}(w)dw - \int_{0}^{1} g(w)f_{X_1|Y_1 = \tilde{y}_2}(w)dw.
\end{multline}
From \eqref{eq:ex-pdf}, for $w \in (1/3,1]$, $f_{X_1|Y_1 = \tilde{y}_1}(w) = f_{X_1|Y_1 = \tilde{y}_2}(w) = 1$. For $w \in [0,1/3]$, $f_{X_1|Y_1 = y_1}(w) = 1 - y_1 + 6wy_1$ and $g(w) = (b^e-1)w$. Equation \eqref{eq:proof-ex-fosd-ranking-eq4} therefore simplifies to:
\begin{align}
\E{g(X_1) | Y_1 = \tilde{y}_1} 
& = \int_{0}^{1/3} g(w)\left(f_{X_1|Y_1 = \tilde{y}_1}(w)-f_{X_1|Y_1 = \tilde{y}_2}(w)\right)dw, \nonumber \\
& = \int_{0}^{1/3} (b^e-1)w\left((1 - \tilde{y}_1 + 6w\tilde{y}_1) - (1 - \tilde{y}_2 + 6w\tilde{y}_2)\right)dw, \nonumber \\
& = \int_{0}^{1/3} (b^e-1)(\tilde{y}_1-\tilde{y}_2)w(6w-1)dw, \nonumber \\
& = (b^e-1)(\tilde{y}_1-\tilde{y}_2)\frac{1}{54} < 0, \label{eq:proof-ex-fosd-ranking-eq5}
\end{align}
where the last inequality is because $\tilde{y}_1 > \tilde{y}_2$ and $b^e < 1$. This establishes \eqref{eq:proof-ex-fosd-ranking-eq1} and the proof is complete. \eop

\section{Proof of Proposition \ref{prop:rev-rank-sp-refinement}} \label{sec:proof-rev-rank-sp-refinement}

\textbf{Proof of part (i)}:

The proof is almost the same as the proof of Proposition \ref{prop:sp-main-result}. The main difference is in how the concluding inequality that ranks the expected payments of the winning buyer under different families of securities is derived using FOSD and strong steepness instead of MLR and steepness. 

As in the proof of Proposition \ref{prop:sp-main-result}, it suffices to show that \eqref{eq:proof-sp-main-result-eq1} holds for any $y_1 > z_1$. The argument in the proof of Proposition \ref{prop:sp-main-result} implies that $\phi(x_1,s(z_1,z_1;\Phi))$ must cross $\psi(x_1,s(z_1,z_1;\Psi))$ as functions of $x_1$. Strong steepness requires that $\phi(x_1,s(z_1 , z_1 ;\Phi)) - \psi(x_1,s(z_1 , z_1 ;\Psi))$ is nondecreasing in $x_1$ and hence $x_1 - \psi(x_1,s(z_1 , z_1 ;\Psi))$ single crosses $x_1 - \phi(x_1,s(z_1, z_1 ;\Phi))$ from below. Inequality \eqref{eq:proof-sp-main-result-eq3} then follows by the same argument as before, implying:
\begin{align*}
& \E{\phi(X_1,s(Z_1, Z_1;\Phi)) - \psi(X_1,s(Z_1, Z_1;\Psi)) | Y_1 = z_1, Z_1 = z_1} \geq 0, \\
\Rightarrow &  \E{\phi(X_1,s(Z_1, Z_1;\Phi)) - \psi(X_1,s(Z_1, Z_1;\Psi)) | Y_1 = y_1, Z_1 = z_1} \geq 0.
\end{align*}
The last inequality that proves the result is from an application of Lemma \ref{lemma:pd-fosd-eq}, using the fact that $\phi(x_1,s(z_1 , z_1 ;\Phi)) - \psi(x_1,s(z_1 , z_1 ;\Psi))$ is nondecreasing in $x_1$ (i.e., strong steepness) together with FOSD. \eop

\vspace*{1\baselineskip} \noindent \textbf{Proof of part (ii)}:

It suffices to consider only the case of two risk neutral buyers and $(\mb{X},\mb{Y})$ satisfying FOSD such that the $(X_i, Y_i)$ pairs for different buyers are independent and identically distributed. We also assume without loss of generality that the $Y_i$'s are distributed over the interval $[0,1].$

Suppose $\Phi$ and $\Psi$ are two families of securities satisfying conditions (a) and (b) of Proposition \ref{prop:rev-rank-sp-refinement}(ii). Let $\wh{b} \in (\underline{b}, \overline{b})$ and $\wt{b} \in (\underline{b}, \overline{b})$ be such that $\phi(w,\wh{b}) - \psi(w,\wt{b})$ assumes both negative and positive values over $ w \in [\underline{x},\overline{x}]$. Let $w_o$ be a point in $[\underline{x},\overline{x}]\backslash E_b,$ and let $\wh{\alpha} \triangleq \pd{\phi(w,\wh{b})}{w}\big\vert_{w=w_o}\big.$  and $\wt{\alpha} \triangleq \pd{\psi(w,\wt{b})}{w}\big\vert_{w=w_o}\big.$. It suffices to prove that $\wh{\alpha} \geq  \wt{\alpha}.$

For $\epsilon \geq 0$, let the function $\rho^{\epsilon}$ be defined by:
\beq{fn-rho}
\rho^{\epsilon}(t) = \left\{
\begin{array}{l l}
-\epsilon^{\frac{1}{3}} & \quad \text{if $-\epsilon^{\frac{1}{3}} \leq t < 0$,}\\
\epsilon^{\frac{1}{3}} & \quad \text{if $0 \leq t \leq \epsilon^{\frac{1}{3}}$,}\\
0 & \quad \text{otherwise.}
\end{array} \right.
\eeq
Note that $\int_{\infty}^{\infty}  t  \rho^{\epsilon}(t) dt = \epsilon.$

Since $\phi(w,\wh{b}) - \psi(w,\wt{b})$ assumes both negative and positive values, there exists a pdf $f_W$ over $[\underline{x}, \overline{x}]$ that is continuously differentiable, strictly positive, and such that if random variable $W$ has this pdf, then $\Ebig{\phi(W,\wh{b})} = \Ebig{\psi(W,\wt{b})}$. Define $r \triangleq \Ebig{\phi(W,\wh{b})} = \Ebig{\psi(W,\wt{b})}$ and let $I = \Ebig{W} - r$. 

We describe joint distributions for $(X_1^{\epsilon}, Y_1)$, parameterized by $\epsilon \geq 0$, using the pdf $f_W$ and the function $\rho^{\epsilon}$. The random variable $Y_1$ is uniformly distributed over the interval $[0,1],$  and
\beq{fx}
f^{\epsilon}_{X_1| Y_1 = y_1}(x_1)  =  f_W(x_1)+ y_1 \rho^{\epsilon}(x_1-w_o).
\eeq
In words, the conditional pdf of $X^{\epsilon}_1$ given $Y_1=y_1$ is obtained from $f_W$ by shifting a small amount of probability mass, proportional to $y_1,$  from just below $w_o$ to just above $w_o.$ We shall only consider $\epsilon$ small enough that the conditional pdf is nonnegative and the rectangular set, (support of $\rho^{\epsilon}(w-w_o)$)$\times$(an open interval containing $\widehat{b}$), is contained in a set of continuous differentiability of $\phi$, and the analogous condition holds for $\widetilde{b}$ and $\psi.$ For each such $\epsilon$, $(\mb{X}^{\epsilon},\mb{Y})$ satisfies FOSD. If $\epsilon=0$, the signals are independent of the values and the pdf of $X_1^{\epsilon}$ is identical to the pdf of $W$. By construction, $\Ebig{W-I-\phi(W,\wh{b})}=0$. Assumption \ref{assumption:ordered-security}(ii) then implies that for $\epsilon = 0$ and any $y$, $\Ebig{X_1^{\epsilon} - I - \phi(X_1^{\epsilon}, \underline{b})|Y_1=y} > 0 $ and $\Ebig{X_1^{\epsilon} - I - \phi(X_1^{\epsilon}, \overline{b})|Y_1=y} < 0 $. By the smoothness conditions in Proposition \ref{prop:rev-rank-sp-refinement}(ii), $\Ebig{X_1^{\epsilon} - I - \phi(X_1^{\epsilon}, b)|Y_1=y}$ is continuous in $\epsilon$. Hence, for small values of $\epsilon$, our choice of $I$ and $(\mb{X}^{\epsilon},\mb{Y})$ satisfy Assumption \ref{assumption:ordered-security}(i) for family $\Phi$. The same holds true for $\Psi$.

Let $R(\epsilon; \Phi)$ denote the expected revenue from the family of securities $\Phi$ for $\epsilon$ parameterized random variables $(\mb{X}^{\epsilon},\mb{Y})$; define $R(\epsilon; \Psi)$ similarly. By the assumed revenue ranking, $R(\epsilon; \Phi)\geq R(\epsilon; \Psi)$ for all $\epsilon$ being considered. If $\epsilon=0$, both buyers bidding~$\wh{b}$ is the symmetric equilibrium for the family of securities $\Phi$, both buyers bidding~$\wt{b}$ is the symmetric equilibrium for the family of securities $\Psi$, and the revenue for each set of securities is~$r$. Hence, $R(0; \Phi) = R(0; \Psi).$ It will be shown below that the derivative of $R(\epsilon; \Phi)$ with respect to $\epsilon$ at zero satisfies $R'(0;\Phi) = (1+\wh{\alpha})/3$. Similarly, $R'(0; \Psi) = (1+\wt{\alpha})/3$. By the revenue ranking, we must have $R'(0; \Phi) \geq R'(0; \Psi)$, implying $\wh{\alpha} \geq \wt{\alpha}$. It remains to show that $R'(0;\Phi) = (1+\wh{\alpha})/3$.

Consider the family $\Phi.$ The hypotheses imply that $\Ebig{\phi(X_1^{\epsilon}, b)|Y_1=y}$ is continuously differentiable in both $\epsilon$ and $b$ for $\epsilon$ in a neighborhood of zero and $b$ in a neighborhood of $\widehat{b}.$ Moreover,
\begin{align}
& \CE{X^{\epsilon}_1}{Y_1 = y_1} = \E{W} + \epsilon y_1, \label{eq:x-eps}\\
& \CE{\phi(X^{\epsilon}_1, b)}{Y_1 = y_1} = r + \epsilon y_ 1 \pd{\phi(w,b)}{w}\bigg\vert_{w=w_o}\bigg. + o(\epsilon), \label{eq:phi-diff}
\end{align}
where \eqref{eq:phi-diff} is obtained by the Taylor series representation of $\phi(w,b)$ centered at $w = w_o$, for a fixed $b$ in a small neighborhood of $\wh{b}$. Here, $\lim_{\epsilon \to 0^+} o(\epsilon)/\epsilon = 0$. Equivalently, 
\beq{d-phi}
\pd{\CE{\phi(X^{\epsilon}_1, b)}{Y_1 = y_1}}{\epsilon}\Bigg\vert_{\epsilon=0}\Bigg. = y_ 1 \pd{\phi(w,b)}{w}\bigg\vert_{w=w_o}\bigg.. 
\eeq
Since the signal and value pairs are independent, the function $s^{\epsilon}(y_1, z_1)$, defined by \eqref{eq:s-sp} for $(\mb{X}^{\epsilon},\mb{Y})$, depends only on $y_1$; we write it is as $s^{\epsilon}(y_1).$  It is characterized by the equation
\beqn
\CE{X_1^{\epsilon} - I - \phi(X_1^{\epsilon}, s^{\epsilon} (y_1) )}{Y_1=y_1} = 0.
\eeqn
For a given $y_1$, \eqref{eq:x-eps}, \eqref{eq:d-phi}, and the smoothness conditions of Proposition \ref{prop:rev-rank-sp-refinement}(ii)(b) imply that the partial derivatives of $\Ebig{X_1^{\epsilon} - I - \phi(X_1^{\epsilon}, b)| Y_1=y_1}$ with respect of $\epsilon$ and $b$ are continuous for~$\epsilon$ in some small interval $[0, \overline{\epsilon})$ and $b$ in a small neighborhood of $\wh{b}$; and the partial derivative with respect to $b$ is nonzero. By the implicit function theorem, $s^{\epsilon}(y_1)$ is differentiable in $\epsilon$ and satisfies:
\beq{diff-s}
\pd{s^{\epsilon}(y_1)}{\epsilon} = 
-\frac{\displaystyle \pd{\Ebig{X_1^{\epsilon} - I - \phi(X_1^{\epsilon}, b)| Y_1=y_1}}{\epsilon}\bigg\vert_{b=s^{\epsilon} (y_1)}\bigg.}{\displaystyle \pd{\Ebig{X_1^{\epsilon} - I - \phi(X_1^{\epsilon}, b)| Y_1=y_1}}{b}\bigg\vert_{b=s^{\epsilon} (y_1)}\bigg.} 
= \frac{y_1 \left(1-\displaystyle \pd{\phi(w,s^{\epsilon} (y_1))}{w}\bigg\vert_{w=w_o}\bigg.\right)}{\displaystyle \pd{\Ebig{\phi(X_1^{\epsilon}, b)| Y_1=y_1}}{b}\bigg\vert_{b=s^{\epsilon} (y_1)}\bigg.}.
\eeq
Notice that at $\epsilon = 0$, $X_1^{\epsilon}$ is independent of $Y_1$, $s^{\epsilon} (y_1) = \wh{b}$ for any $y_1$, and $X_1^{\epsilon}$ and $W$ are identical in distribution.
For notational convenience, define $D \triangleq \pd{\E{\phi(W, b)}}{b} \big\vert_{b = \wh{b}}\big.$. Then from \eqref{eq:d-phi}, \eqref{eq:diff-s} and continuity of derivatives,
\beq{diff-s0}
\pd{s^{\epsilon}(y_1)}{\epsilon}\bigg\vert_{\epsilon = 0}\bigg. = \frac{y_1(1-\wh{\alpha})}{D},
\eeq
and
\begin{align}
\pd{\CE{\phi(X_1^{\epsilon}, s^{\epsilon}(y_2))}{Y_1=y_1}}{\epsilon}\bigg\vert_{\epsilon = 0}\bigg. 
& = \pd{\Ebig{\phi(X_1^{\epsilon}, \wh{b}) \vert Y_1=y_1}}{\epsilon}\bigg\vert_{\epsilon = 0}\bigg. + \left(\pd{\E{\phi(W, b)}}{b} \pd{s^{\epsilon}(y_2)}{\epsilon}\right)\bigg\vert_{\epsilon = 0, b = \wh{b}}\bigg.,  \nonumber \\
& = \wh{\alpha} y_1 + (1-\wh{\alpha})y_2. \label{eq:phi-eps}
\end{align}

Next, notice that 
\beq{rev-eq}
R(\epsilon; \Phi) = \CE{\phi(X_1^{\epsilon} , s^{\epsilon}(Y_2) )}{Y_1 > Y_2 } = \CE{\CE{\phi(X_1^{\epsilon} , s^{\epsilon}(Y_2) )}{Y_1, Y_2}}{Y_1 > Y_2 }.
\eeq
Since $\Ebig{\phi(X_1^{\epsilon} , b)|Y_1, Y_2}$ is continuously differentiable in $\epsilon$ and $b$, and $s^{\epsilon}(Y_2)$ is differentiable in $\epsilon$, $\Ebig{\phi(X_1^{\epsilon} , s^{\epsilon}(Y_2) )|Y_1, Y_2}$ is continuously differentiable in $\epsilon$. Additionally, because $Y_i$'s take values in finite interval $[0,1]$, in order to compute the derivative of $R(\epsilon; \Phi)$, we can take the derivative inside the outer expectation in \eqref{eq:rev-eq}.
\begin{align*}
\pd{R(\epsilon; \Phi)}{\epsilon}\bigg\vert_{\epsilon=0}\bigg. 
& = \E{\pd{\CE{\phi(X_1^{\epsilon} , s^{\epsilon}(Y_2) )}{Y_1, Y_2}}{\epsilon}\Bigg\vert_{\epsilon=0}\bigg. \Bigg\vert Y_1 > Y_2 \bigg.}, \\
& = \CE{\wh{\alpha} Y_1 + (1-\wh{\alpha})Y_2}{Y_1 > Y_2}, \\
& = \frac{1+\wh{\alpha}}{3},
\end{align*}
where the second equality follows from \eqref{eq:phi-eps}. Therefore, $R'(0;\Phi) = (1+\wh{\alpha})/3,$ as required. \eop

\bibliographystyle{abbrvnatold}
\bibliography{../../../../Citations/AuctionTheory,../../../../Citations/Misc}

\end{document}